\documentclass[runningheads]{llncs}

\usepackage{amsmath}
\usepackage{algorithm}
\usepackage[noend]{algpseudocode}
\usepackage[bottom]{footmisc}
\usepackage{tikz}
\makeatletter
\def\BState{\State\hskip-\ALG@thistlm}
\makeatother

\usepackage{graphicx}

\usepackage{algorithm} 
\usepackage{cite}
\usepackage{algorithmicx}
\usepackage[noend]{algpseudocode}
\usepackage[bottom]{footmisc}
\usepackage{tikz}
\newcommand\extralabel[4][0mm]{\node[label={[label distance=#1]#2:#3}] at (#4){};}
\usetikzlibrary{shapes.arrows, fadings}
\usetikzlibrary{decorations.text}
\usepackage{multicol}
\usepackage{mathtools}
\usepackage{adjustbox}
\usepackage{enumitem}
\usepackage{diagbox}
\usepackage{array}
\newcolumntype{P}[1]{>{\centering\arraybackslash}p{#1}}
\newcolumntype{M}[1]{>{\centering\arraybackslash}m{#1}}

\makeatletter
\def\BState{\State\hskip-\ALG@thistlm}
\makeatother
\usetikzlibrary{calc, chains,
                fit,
                positioning,
                shapes}

\definecolor{myblue}{RGB}{80,80,160}
\definecolor{mygreen}{RGB}{80,160,80}
\definecolor{myempty}{RGB}{255,255,255}

\usetikzlibrary{shapes,arrows,positioning,decorations.markings}

\usepackage{graphicx}
\usepackage{hyperref}
\hypersetup{
    colorlinks=true,
    linkcolor=blue,
    filecolor=magenta,      
    urlcolor=cyan,
}

\usepackage{amsmath}

\usepackage{amsthm}

\usepackage{ upgreek }

\begin{document}

\title{Distributed $K$-Backup Placement \\
and Applications to Virtual Memory \\ in Real-World Wireless Networks}

\titlerunning{Distributed $K$-BP
and Applications to VM in Real-World Wireless Networks}

\author{Gal Oren\inst{1,2} 
\and Leonid Barenboim\inst{3}}

\authorrunning{G. Oren, L. Barenboim}

\institute{
Department of Computer Science, Ben-Gurion University of the Negev, P.O.B. 653, Be'er Sheva, Israel 
\and Department of Physics, Nuclear Research Center-Negev, P.O.B. 9001, Be'er-Sheva, Israel
\and Department of Mathematics and Computer Science, The Open University of Israel, P.O.B. 808, Ra'anana, Israel\\
\email{orenw@post.bgu.ac.il, leonidb@openu.ac.il}}
\maketitle              

\begin{abstract}
The \textit{Backup Placement problem} in networks in the $\mathcal{CONGEST}$ distributed setting considers a network graph $G = (V,E)$, in which the goal of each vertex $v \in V$ is selecting a neighbor, such that the maximum number of vertices in $V$ that select the same vertex is minimized \cite{halldorsson2015bp}. Previous backup placement algorithms suffer from obliviousness to main factors of real-world heterogeneous wireless network. Specifically, there is no consideration of the nodes memory and storage capacities, and no reference to a case in which nodes have different energy capacity, and thus can leave (or join) the network at any time. These parameters are strongly correlated in wireless networks, as the load on different parts of the network can differ greatly, thus requiring more communication, energy, memory and storage. 
In order to fit the real-world attributes of wireless networks, this work addresses a generalized version of the original problem, namely \textit{$K$-Backup Placement}, in which each vertex selects $K$ neighbors, for a positive parameter $K$. Our \textit{$K$-Backup Placement} algorithm terminates within just one round.
In addition we suggest two complementary algorithms which employ \textit{$K$-Backup-Placement} to obtain efficient virtual memory schemes for wireless networks. The first algorithm divides the memory of each node to many small parts. Each vertex is assigned the memories of a large subset of its neighbors. Thus more memory capacity for more vertices is gained, but with much fragmentation. The second algorithm requires greater round-complexity, but produces larger virtual memory for each vertex without any fragmentation.

\end{abstract}

\newpage
\section{Introduction}

\subsection{Fault-Tolerance and Data Loss in IoT and WSNs}

IoT (Internet-of-Things) devices involve very differentiated requirements in terms of communication fashion, computation speed, memory limit, transmission rate and data storage capacity \cite{aggarwal2013internet}. Considering these devices as nodes in a wireless network (forming a Wireless Sensor Network, or WSN) they experience constant failures due to environmental factors, battery exhaustion, damaged communications links, data collision, wear-out of memory and storage units and overloaded sensors \cite{kakamanshadi2015survey}. As massive volume of heterogeneous data is created by those diverse devices periodically, it needs to be sent to other locations. Moreover, As those devices may operate in a nonstop fashion and usually placed in geographically diverse locations, they introduce a challenge in terms of communication overhead and storage capacity, since transfer and storage of data are resource-consuming activities within an IoT data management \cite{abu2013data}. Hence, these factors are crucial in data management techniques for IoT and WSNs.

The challenging situation can be described as follows. As data is aggregated in a device or in the concentration storage points within the IoT network, the amount of vacant storage becomes very limited. Thus the data usually needs to be sent further over the network, either to a storage relay node, or storage end-node \cite{fan2010scheme}. As wireless broadband communications is most likely to be used all along this process in the IoT world, there are also strict limitations to the amount of communication each power-limited device can perform. Moreover, since crashes of nodes, failures of communication links, and missing data are unavoidable in wireless networks, fault-tolerance becomes a key-issue. Among the causes of these constant failures are environmental factors, damaged communications links, data collision and overloaded nodes \cite{kakamanshadi2015survey}. Thus, the problems of efficient data transfer, storage and backup, become major topics of interest in IoT \cite{qin2016things}.

\subsection{The Backup Placement Problem in Networks}

The {\em Distributed Backup Placement} problem was introduced by  Halldorsson, Kohler, Patt-Shamir, and Rawitz \cite{halldorsson2015bp} in 2015. It is very well motivated by networks whose nodes may have memory faults, and wish to store backup copies of their data at neighboring nodes \cite{oren2018distributed}. Unfortunately, neighboring nodes may incur faults as well, and so the number of nodes that select the same node as their backup should be minimized. This way, if a backup node incurs fault, the number of nodes in the network that lose data is minimized. Moreover, a proper backup placement allows each vertex to perform a backup to a neighboring node, rather than a more distant destination, and thus improves network performance. In addition, nodes' memories are used to the minimum extent for the purpose of backups, which makes it possible to maximize the memory available for other purposes of the vertices. 

Throughout this paper, we consider the $\mathcal{CONGEST}$ distributed model, In this model a network is represented by a graph $G = (V,E)$, whose vertices ($V$, $|V|=n$) represent processors and edges ($E \subseteq V \times V$, $|E|=m$) represent communication links between pairs of processors. Time proceeds in discrete synchronous rounds, where in each round vertices receive and send  message to their neighbors (i.e. those processors with which they share a communication link), and perform local computations. The bandwidth on each communication line on each round is bounded by $O(\log{n})$. 
The network graph is considered as the input, where initially each vertex knows only its own ID and the IDs of its neighbors. The goal of each vertex is computing its part for a certain problem. The neighbors set of each node $v \in V$ is denoted by $\Gamma(v)$. The degree of a node $v \in V$ is $deg(v) = |\Gamma(v)|$. The maximum (respectively, minimum) degree of graph $G$ is $\Delta(G) = \max{(deg(v))}$ (res., $\delta(G) = \min{(deg(v))}$) for any $v \in V$. The distance $dist(u,v)$ between two vertices $u,v \in V$ equals to the length of the shortest path connecting $u$ and $v$ in $G$.

The neighborhood independence of graph $G=(V,E)$ is the maximum number of independent neighbors a vertex in the graph has. Formally, Neighborhood independence $I(G)$ of a graph $G=(V,E)$ is the maximum size of an independent set contained in a neighborhood $\Gamma(v), v \in V$. The notion of neighborhood independence was introduced in \cite{barenboim2011deterministic} and has been intensively studied since then \cite{barenboim2017deterministic, barenboim2018distributed, kuhn2018deterministic, assadi2019algorithms}.

A function $\varphi :V \rightarrow [\alpha]$ is a legal \textit{$\alpha$-Coloring} of graph $G=(V,E)$ if and only if, for each $\{v,u\}\in E \rightarrow \varphi(v) \neq \varphi(u)$. A subset of vertices assigned to the same color is called a \textit{color class}. Note that every such class forms an \textit{independent set}. Thus, a $k$-coloring is the same as a partition of the vertex set into $k$ independent sets. We define \textit{super-class} to include a range of color classes. A function $\varphi' :V \rightarrow [\alpha]$ is a legal \textit{$t$-hop $\alpha$-Coloring} of graph $G=(V,E)$ if and only if, for each $u,v \in V$ such that $dist(u,v) \leq t$, $\varphi'(v) \neq \varphi'(u)$.

In this paper we generalize the classic backup placement problem. The generalized version is called {\em $K$-Backup Placement}. It is defined as follows. 

\begin{definition}[$K$-Backup Placement problem]
Given an unweighted unoriented graph $G = (V,E)$ representing the network, and a constant $K \geq \delta(G)$, each vertex $v_i \in V$ ($1 \leq i \leq n$) selects $K$ neighbors $B_i=\{v_i^1, v_i^2, ..., v_i^K\}$. For any vertex $u_j \in V$ ($1 \leq j \leq n$), let $b_j$ be the number of vertices performed backup on node $u_j$. The goal in the $K$-Backup Placement problem is to minimize the $b_j$ (for any $1 \leq j \leq n$).
\end{definition}

Note that the classic Backup Placement problem is a special case of $K$-Backup Placement where $K=1$. 
Several solutions to the 1-Backup Placement problem were introduced over the last decade. Halldorsson et al. \cite{halldorsson2015bp}, who presented the problem, obtained an $O(\log n/ \log \log n)$ approximation with randomized polylogarithmic time. Their algorithm remained the state-of-the-art for general graphs, as well as several specific graph families. Barenboim and Oren \cite{barenboim2020fast} obtained significantly improved algorithms for various graph topologies. Specifically, they showed that an $O(1)$-approximation to optimal 1-Backup Placement can be computed {\em deterministically} in $O(1)$ rounds in wireless networks, and more generally, in any graph with {\em neighborhood independence bounded by a constant}. At the other end, they considered sparse graphs, such as trees, forests, planar graphs and graphs of constant arboricity\footnote[1]{ {\em Arboricity} is the minimum number of forests that the graph edges can be partitioned into or the maximum ratio of edges to nodes in any subgraph. The arboricity of a graph is a measure for its sparsity. Sparse graphs have low arboricity.}, and obtained constant approximation to optimal 1-Backup Placement in $O(\log n)$ deterministic rounds. They also considered two variant of networks, specifically faultless networks and faulty networks. For the latter, they obtained a self-stabilizing algorithm and proved stabilization within 1 to 3 rounds \cite{barenboim2020distributed}. 

These algorithms are extremely efficient for wireless networks models as commonly the WSN is modeled by $n$ nodes, each of which has the same communication range (of 1 unit) and memory size (1 unit). Within a single round all nodes can communicate with other nodes in their communication range. Each node can send a message of restricted size to each of the nodes in its communication range (Henceforth, neighboring nodes.) The nodes goal is to sense and process some data, that change periodically. The correlation between the $K$-Backup Placement problem and WSNs arise from the fact that each node may be required to process more data than its own memory and storage capacity. Consequently, excessive data has to be sent to neighboring nodes, without overloading them. (Recall that the degree of each vertex is at most $\Delta$.)

However, in practice, real-world attributes of the network and its nodes should be taken into account. Therefore, those algorithms should be extended to include other important aspects \cite{oren2018distributed}, such as:

\begin{itemize}
  \item \textbf{Memory and Storage}: As the amount of RAM and Storage of the node is final and non-uniform \cite{kakamanshadi2015survey} - either because the nodes themselves differ in those parameters or because the nodes joined the network in non-synchronized periods - backup placement cannot be oblivious to it. Moreover, because the RAM acts as a buffer towards the flash memory \cite{karray2018comprehensive}, the amount of pages is limited by the RAM capacity, which is usually about an order of magnitude smaller than the flash capacity. Therefore, the size of the packets cannot exceed the RAM capacity during the backup operation \cite{wagner2007algorithms}.
  \item \textbf{Communication and Energy}: The amount of available energy of a node - as well as the amount of communication rounds the node can perform - is final and non-uniform \cite{karray2018comprehensive}. This is either because the nodes themselves can hold different battery capabilities and communication range; or because the nodes joined the network in non-synchronized periods; or performed in different roles which require different amount of energy \cite{tuna2013energy} - backup placement cannot be oblivious to it.
  \item \textbf{Nodes that Join or Leave}: The backup placement algorithm cannot be oblivious to premature death of nodes and joining of new ones to the network \cite{rostami2018survey}. This is because we cannot assume that the network is able to make appropriate cluster adaptation (i.e. triggered to re-cluster in these changes of the topology) nor in the capability of the backup placement algorithm to address those changes without loosing backups or missing the new opportunity to backup data to new joined nodes \cite{kong2013data}.
\end{itemize}

In order to solve the problem of backup placement in real-world wireless networks as described above, we devise a $K$-Backup Placement algorithm. Indeed, if a vertex selects $K$ neighbors, even if some of them crash, there is still a possibility to employ the neighbors that remain. We stress that the desired procedure should compute $O(1)$-approximate backup placement in graphs with \textit{constant neighborhood independence $c=I(G)$}. This family of graphs includes many wireless sensor network topologies such as, UDG, QUDG, UBG, GBG, etc\footnote[2]{A graph is \textit{growth-bounded} (GBG) or \textit{independence-bounded} (BIG) if the number of independent nodes in a node’s $r$-neighborhood is bounded. Intuitively, if many nodes are located close from each other, many of them must be within mutual transmission range. As the model only restricts the number of independent nodes in each neighborhood it is therefore a generalization of the \textit{unit disk} graph (UDG), the \textit{quasi unit disk} graph (QUDG) and the \textit{unit ball} graph (UBG) models, which are essential model of WSNs \cite{schmid2008modeling}. We stress that for a variety of real-life network topologies, the neighborhood independence $c$ is bounded by a small constant. For example, let $G$ be a UDG, the neighborhood independence of a vertex $u$ is at most $c=5$.}.
The procedure receives a graph $G = (V,E)$ as input, and all the selections should be performed in parallel within a single round. Indeed, in this setting we obtain an $O(1)$-approximate $K$-Backup Placement within just a single round.

Once we have a solution for the $K$-Backup Placement problem, we turn to a more challenging task of allocating virtual memory for nodes in wireless networks. Here, the goal is partitioning the node set into classes, and obtain a scheduling in which each class is active at a time and can make use of memories of other classes that are inactive in that time. The challenge here is coming up with the right trade-off between the number of classes (that needs to be small) and the size of the virtual memory (that becomes larger as the number of classes increase). To obtain the desired balance we compute a certain coloring in which all vertices of the same color class have considerable amount of virtual memory, while keeping the number of classes sufficiently small.

All of our algorithms are {\em deterministic}. The round complexity of our algorithms is $O({\log}^*{n} + K ^ 4)$ for graphs with constant neighborhood independence $c$ (i.e., dense graphs) and $O({\log}^*{n} + R)$ for graphs with bounded growth, where $R$ is a positive parameter (rate of data change, $R < \Delta$). The first algorithm achieves a memory increase by a multiplicative factor $O(K)$ per active vertex, and the second algorithm achieves a memory increase by a multiplicative factor $O(R)$. We stress that the $K$-Backup Placement algorithms in this paper optimizes memory and storage backups in the entire network, and as such inherently reduce communication, which consequently reduce energy consumption, which eventually decreases the amount of nodes that leave the network. Moreover, the distributed fashion of the algorithms, which relies on the neighborhood of each vertex solely, allow us to adjust to network changes. Thus, our solutions suit a large family of real-world networks.

\begin{table}[H]
    \begin{center}
    \begin{adjustbox}{width=1\textwidth}
    \small
        \begin{tabular}{p{4cm}|p{3cm}|p{3cm}|p{3cm}|p{1.5cm}|}
            Problem / Topology & Very dense: Graphs of bounded growth & Dense: Graphs of constant neighborhood independence & Sparse: Graphs of constant arboricity & General Graphs \\
            \hline
            $O(\frac{\log n}{\log \log n})$-approximate \newline 1-Backup Placement  & \multicolumn{4}{c|}{$O(\frac{\log^ 6 n}{\log^4 \log n}$) \ \ \ (randomized)  \ \ \cite{halldorsson2015bp} } \\ 
            \hline
            $O(1)$-approximate \newline 1-Backup Placement & \multicolumn{2}{c|}{$O(1)$ \cite{barenboim2020fast} \cite{barenboim2020distributed}}  & $O(\log n)$, $\Omega(\sqrt {\frac{\log n}{\log \log n}})$  \cite{barenboim2020fast} \cite{barenboim2020distributed} & \\
            \hline
            $O(K)$-approximate \newline $K$-Backup Placement & \multicolumn{2}{c|}{$O(1)$ 
            \textbf{This paper}} & & \\
            \hline
        \end{tabular}
        \end{adjustbox}
        \caption{Running times comparison of state-of-the-art backup placement algorithms.}
        \label{my-label}
    \end{center}
\end{table}

\vspace{-1cm}

\section{\textsc{$K$-Backup Placement} by \textsc{$K$-Next-Modulo} Algorithm}
Before we introduce the algorithm for $K$-Backup Placement, we would like to introduce the \textsc{$K$-Next-Modulo} algorithm which will be later used as a building-block. This is a generalization of \cite{barenboim2020distributed}. In \cite{barenboim2020distributed} Barenboim and Oren defined an operation named \textsc{Next-Modulo} that receives a vertex $v$ and a set of its neighbors $\Gamma(v)$ in the graph $G$. The operation \textsc{Next-Modulo}($v$, $\Gamma(v)$), selects a vertex in $\Gamma(v)$ with a higher ID than the ID of $v$, and whose ID is the closets to that of $v$. If no such neighbor is found, then the operation returns the neighbor with the smallest ID. All these selections are performed in parallel within a single round. This completes the description of the algorithm. In the current paper we generalize the \textsc{Next-Modulo} algorithm. The motivation behind this new algorithm is that a usage of a generalized version of the \textsc{Next-Modulo} function \cite{barenboim2020distributed} in the selection process of the backup placement can effectively bound the load on all of the vertices while keeping good performances.

Given a graph $G=(V,E)$ with bounded \textit{neighborhood independence} \textit{c}, we define an operation \textsc{$K$-Next-Modulo} that receives a vertex $v$, a set of its neighbors $\Gamma(v)$ in the graph $G$, and the parameter $K$. The operation \textsc{$K$-Next-Modulo}($v$, $\Gamma(v)$, $K$), selects the $K$ neighbors that immediately succeed $v$, in a circular ordering according to vertex IDs. 
Formally, denote $d = \deg(v) + 1$, and let $u_1,u_2,...,u_i = v, u_{i + 1},..., u_{\deg(v)+1}$ be an ordering of the neighborhood of $v$, according to IDs, in ascending order. Then $v$ selects the $K$ neighbors:
$u_{(i+1) \: \mbox{mod} \: d}, u_{(i+2) \: \mbox{mod} \: d},..., u_{(i+K) \: \mbox{mod} \: d}$.
All these selections are performed in parallel within a single round. This completes the description of the algorithm. Its pseudocode is provided in Algorithm \ref{nextmoduleka}. Theorem \ref{theorem:k-next-modulo-correctness} summarizes its correctness.

\begin{algorithm}[H]
\caption{\textsc{$K$-Next-Modulo} Algorithm}
\label{nextmoduleka}
\begin{algorithmic}[1]
\Procedure{$K$-Next-Modulo($v$, $\Gamma(v)$, $K$)}{}

\State \textbf{sort} $\Gamma(v)$ by IDs as a \textit{Circular Linked List}
\State $v_{index} = \Gamma(v)[0]$

\State \textbf{foreach} $v_{neighbor} \in \Gamma(v)$
\State \hspace{0.5cm} \textbf{if} $v_{neighbor}.ID > v.ID$ 
\State \hspace{1.0cm} $v_{index} = v_{neighbor}$
\State \hspace{1.0cm} \textbf{break} \Comment{Selects a vertex in $\Gamma(v)$ with a higher ID than the ID of $v$, and whose ID is the closets to that of $v$. If no such neighbor is found, then the operation returns the neighbor with the smallest ID.}

\State \textbf{while} $K>0$ \textbf{do}
\State \hspace{0.5cm} \textbf{if} $v_{index} \notin v.BP$ \textbf{then}
\State \hspace{1.0cm} $v.BP.\textbf{append}(v_{index})$ \Comment{Find the $K$ vertices next to $v$ in $\Gamma(v)$, and choose it.}
\State \hspace{0.5cm} $K = K-1$
\State \hspace{0.5cm} $v_{index} = v_{index}.\textbf{next}$
\State \textbf{return} $v$

\EndProcedure
\end{algorithmic}
\end{algorithm}

\begin{theorem}
\label{theorem:k-next-modulo-correctness}
In a graph $G = (V, E)$ with neighborhood-independence bounded by a constant $c$, and with minimum degree at least $K$, any vertex $v \in V$ is selected by at most $c \cdot K$ vertices.
\end{theorem}

\begin{proof}
Assume for contradiction that more than $c \cdot K$ neighbors chose $v$. Denote by $U$ the subset of neighbors that selected $v$. Let $w_1$ be the vertex that immediately succeeds $v$ in a circular ordering according to IDs of $U \cup \{v\}$. It follows that $w_1$ is connected by edges to at most $K - 1$ vertices of $U \setminus \{w_1\}$. Otherwise, there are $k$ neighbors of $w_1$ that succeed $w_1$ and precede $v$ in the circular ordering, and thus $v$ would not have been selected by $w_1$. We remove $w_1$ and its neighbors from $U$. Next, we find a remaining vertex $w_2$ in $U$ that immediately succeeds $v$ in a circular ordering according to IDs of remaining vertices in $U \cup \{ v \}$. Similarly, this vertex can be connected to at most $K - 1$ vertices in $U$. We remove $w_2$ and its neighbors from $U$. We repeat this for $i =1,2,...,c+1$ iterations. (See illustration in Figure \ref{fig_proof}.) The number of neighbors in each iteration becomes smaller by at most $k$ neighbors. The initial number of neighbors is at least $c \cdot K + 1$. Therefore, it is possible to execute $c + 1$ iterations. In each iteration $i = 1,2,...,c + 1$ we remove a vertex $w_i$ and its neighbors from $U$. Hence the set $w_1,w_2,...,w_{c + 1}$ is an independent set of neighbors of $v$. This is a contradiction to the neighborhood independence that is bounded by c. 

\newpage
\vspace{-0.6cm}
\begin{figure}[H]
\centering
\begin{tikzpicture}

  \tikzstyle{vertex}=[circle,fill=black!25,minimum size=12pt,inner sep=2pt]
  
  \tikzstyle{vertex_v}=[circle,fill=red!20,minimum size=12pt,inner sep=2pt]
  
  \node[vertex] (G_1) at (1,0) {1};
  \node[vertex] (G_2) at (-1,0) {2};
  \node[vertex] (G_3) at (1,-2) {5};
  \node[vertex_v] (G_4) at (-1,-2) {4}; \extralabel{-135}{$v$}{G_4};
  \node[vertex] (G_5) at (-2.5,-1) {3};
  \node[vertex] (G_6) at (2.5,-1) {6};
  
  \draw [->,line width=0.2mm, black!30] (G_1) -- (G_2); 
  \draw [->,line width=0.2mm, black!30] (G_2) -- (G_5);
  \draw [->,line width=0.2mm, black!30] (G_5) -- (G_4);
  \draw [->,line width=0.2mm, black!30] (G_4) -- (G_3);
  \draw [->,line width=0.2mm, black!30] (G_3) -- (G_6);
  \draw [->,line width=0.2mm, black!30] (G_6) -- (G_1);
  
\draw [->, line width=0.2mm, blue] [dashed] (G_3) to[out=135,in=35] (G_4);
\draw [->, line width=0.2mm, blue] [dashed] (G_5) to[out=0,in=125] (G_4);
\draw [->, line width=0.2mm, blue] [dashed] (G_2) -- (G_4);
\draw [->, line width=0.2mm, blue] [dashed] (G_1) to[out=225,in=55] (G_4);
\draw [->, line width=0.2mm, blue] [dashed] (G_6) -- (G_4);

\end{tikzpicture}
\end{figure}

\vspace{-1.5cm}
\begin{figure}[H]
\centering
\begin{tikzpicture}[every node/.append style={draw=gray, left color=white, single arrow}]
\node at (0,-1) [
    right color=black!50,
    single arrow,
    minimum height=1cm,
    shading angle=0,
    rotate=270
] {};
\end{tikzpicture}
\end{figure}
\vspace{-1.5cm}
\begin{figure}[H]
\centering
\begin{tikzpicture}

  \tikzstyle{vertex}=[circle,fill=black!25,minimum size=12pt,inner sep=2pt]
  
  \tikzstyle{vertex_v}=[circle,fill=red!20,minimum size=12pt,inner sep=2pt]
  
  \tikzstyle{vertex_selected}=[circle,fill=blue!15,minimum size=12pt,inner sep=2pt]
  
  \node[vertex_v] (G_1) at (1,0) {1};
  \node[vertex] (G_2) at (-1,0) {2};
  \node[vertex_selected] (G_3) at (1,-2) {5}; \extralabel{-45}{$w_1$}{G_3};
  \node[vertex_v] (G_4) at (-1,-2) {4}; \extralabel{-135}{$v$}{G_4};
  \node[vertex] (G_5) at (-2.5,-1) {3};
  \node[vertex_v] (G_6) at (2.5,-1) {6};
  
  \draw [->,line width=0.2mm, black!30] (G_1) -- (G_2); 
  \draw [->,line width=0.2mm, black!30] (G_2) -- (G_5);
  \draw [->,line width=0.2mm, black!30] (G_5) -- (G_4);
  \draw [->,line width=0.2mm, black!30] (G_4) -- (G_3);
  \draw [->,line width=0.2mm, black!30] (G_3) -- (G_6);
  \draw [->,line width=0.2mm, black!30] (G_6) -- (G_1);

  \draw [->, line width=0.2mm, blue] [dashed] (G_3) to[out=135,in=45] (G_4);
  \draw [->, line width=0.2mm, blue] [dashed] (G_3) to[out=60,in=180] (G_6);
  \draw [->, line width=0.2mm, blue] [dashed] (G_3) -- (G_1);

\end{tikzpicture}
\end{figure}
\vspace{-1.5cm}
\begin{figure}[H]
\centering
\begin{tikzpicture}[every node/.append style={draw=gray, left color=white, single arrow}]
\node at (0,-1) [
    right color=black!50,
    single arrow,
    minimum height=1cm,
    shading angle=0,
    rotate=270
] {};
\end{tikzpicture}
\end{figure}
\vspace{-1.5cm}
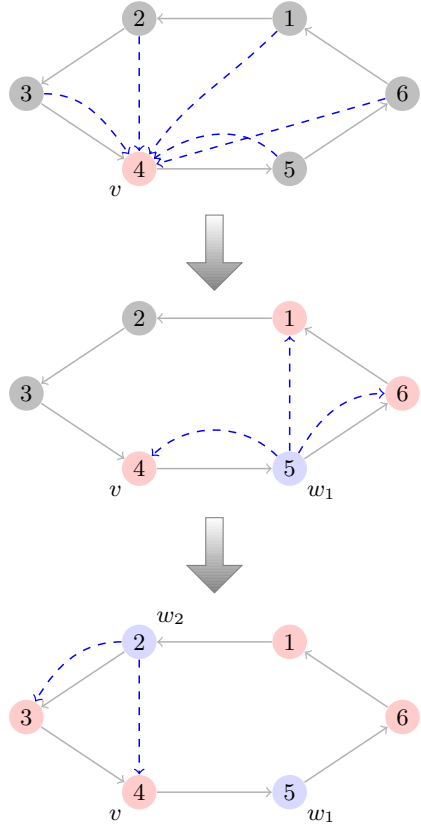
\begin{figure}[H]
\centering
\begin{tikzpicture}

  \tikzstyle{vertex}=[circle,fill=black!25,minimum size=12pt,inner sep=2pt]
  
  \tikzstyle{vertex_v}=[circle,fill=red!20,minimum size=12pt,inner sep=2pt]
  
  \tikzstyle{vertex_selected}=[circle,fill=blue!15,minimum size=12pt,inner sep=2pt]
  
  \tikzstyle{vertex_selected_2}=[circle,fill=yellow!15,minimum size=12pt,inner sep=2pt]
  
  \node[vertex_v] (G_1) at (1,0) {1};
  \node[vertex_selected] (G_2) at (-1,0) {2}; \extralabel{45}{$w_2$}{G_2};
  \node[vertex_selected] (G_3) at (1,-2) {5}; \extralabel{-45}{$w_1$}{G_3};
  \node[vertex_v] (G_4) at (-1,-2) {4}; \extralabel{-135}{$v$}{G_4};
  \node[vertex_v] (G_5) at (-2.5,-1) {3};
  \node[vertex_v] (G_6) at (2.5,-1) {6};
  
  \draw [->,line width=0.2mm, black!30] (G_1) -- (G_2); 
  \draw [->,line width=0.2mm, black!30] (G_2) -- (G_5);
  \draw [->,line width=0.2mm, black!30] (G_5) -- (G_4);
  \draw [->,line width=0.2mm, black!30] (G_4) -- (G_3);
  \draw [->,line width=0.2mm, black!30] (G_3) -- (G_6);
  \draw [->,line width=0.2mm, black!30] (G_6) -- (G_1);

  \draw [->, line width=0.2mm, blue] [dashed] (G_2) -- (G_4);
  \draw [->, line width=0.2mm, blue] [dashed] (G_2) to[out=180,in=60] (G_5);

\end{tikzpicture}

\caption{Exemplifying the assumption for contradiction that more than $c \cdot K$ neighbors chose $v$ in $C_6$ cycle graph (some edges were removed for clarity). From top to bottom: We first assume for contradiction that all vertices select $v = 4$ (marked in red), while $K=3$; then, $v.next= 5$ ($w_1$, marked in blue) forced to select $v$ and another $K-1$ vertices; then, $v.next = 2$ ($w_2$, marked in blue) finishes available selections.} 
\label{fig_proof}
\end{figure}
\vspace{-1.2cm}
\end{proof}

\section{Efficient Virtual Memory by Color-Classes}

The $K$-\textsc{Next-Modulo} algorithm can be considered as a way for division of neighbors to action-groups while keeping the local load balanced. In particular, each of the $K$ neighbors can be selected as backup node of the node performing the algorithm. While the original \textsc{Next-Modulo} algorithm of \cite{barenboim2020distributed} was proved as a valuable building block for backup placement, the $K$-\textsc{Next-Modulo} algorithm turns out to be even more useful for obtaining virtual memory in wireless networks. While in the original algorithm each vertex selects just one neighbor whose memory can be used, our algorithm can be used such that each vertex selects $K$ neighbors in order to use their memories. However, we would like to prevent a situation where multiple vertices use the memory of the same neighbor during the same round. In such a case, for a graph with bounded neighborhood independence $c$, instead of gaining an increase by a factor of $K$, we may gain as low factor as $K/(c \cdot K)$, which is not desirable. To avoid it, we should prevent the possibility that two vertices use the memory of the same neighbor at the same time. To this end, we define the graph $G' = (V,E')$, where $E'$ is the set of edges selected during the \textsc{$K$-Next-Modulo} algorithm. Since each vertex selects $K$ neighbors, and is selected by $O(c \cdot K)$ neighbors, the maximum degree of $G'$ is $\Delta' = O(c \cdot K + K) = O(c \cdot K)$. 

Next, we devise a new scheduling algorithm which first invokes the \textsc{$K$-Next-Modulo} for computing $K$-Backup Placement, and afterwards, the resulting subgraph $G'$ is colored by a 2-hop coloring in order to divide the vertices according to color classes, each of which is activated in a distinct round. Since the subgraph $G'$ is considered to be a general graph, and as we need to perform a 2-hop coloring, we invoke Linial algorithm for $O({\Delta} ^ 2)$-coloring \cite{linial1987distributive} \cite{linial1992locality} on $G'^2$ graph\footnote[3]{Given a graph $G=(V,E)$, $G^2=(V,\hat{E})$ where $\{u,v\}\in \hat{E}$ i.f.f. $\exists w \in V$ s.t. $\{u,w\},\{w,v\} \in E$.}
that produces $O({\Delta'} ^ 4)$ colors with round complexity of ${\log}^*{n} + O(1)$. Yet, as the subgraph $G'$ has maximal degree of $O(c \cdot K)$, while $c$ is a small constant and $K$ is a parameter that can be sufficiently small, the number of colors is relatively small, which will help us to keep good performances in the backup placement process. Moreover, we stress that $K$ is de-facto represents a selective tradeoff between memory and running time, since as $K$ is larger, more memory would be dedicated for backup placement in the expense of the running time, and vice versa. The algorithm pseudocode is provided in Algorithm \ref{algo-knm}. The correctness of the algorithm is given by Theorem \ref{theorem:k-backup}.

\begin{algorithm}[ht]
\caption{Efficient Virtual Memory by Color-Classes by \textsc{$K$-Next-Modulo}}
\label{algo-knm}
\begin{algorithmic}[1]
\Procedure{Efficient-VM(Graph $G = (V,E)$, $K$)}{}
\State {\bf foreach} active node $v \in G$ in parallel do:
\State \hspace{0.5cm} v.BP = \textsc{$K$-Next-Modulo}($v$, $\Gamma(v)$, $K$) \Comment{Active vertices select $K$ neighbors}
\State {\bf foreach} active node $v \in G$ in parallel do:
\State \hspace{0.5cm} \textbf{perform} 2-hop coloring with $O({\Delta'} ^ 2)$ colors \Comment{Divide active vertices to apply turns}
\State \textbf{foreach} \textit{cc} $\in$ \textit{Color-Classes} do: \Comment{Each vertex knows its Color-Class}
\State \hspace{0.5cm} {\bf foreach} node $v \in G'(cc)$ in parallel do: \Comment{Round-robin for Color-Classes vertices}
\State \hspace{1.0cm} \textbf{distribute} the backup placement in parallel from each $v \in G'(cc)$ $\rightarrow$ $v.BP$ vertices
\EndProcedure
\end{algorithmic}
\end{algorithm}

\begin{lemma}
\label{lemma:k-backup-mem}
The virtual memory of each active vertex in $V$ is increased by $K$.
\end{lemma}
\begin{proof}
In each round only a single color class, out of $O(K^4)$, is activated. Since the color classes belong to a 2-hop coloring of $G'$, for each vertex $v$ in such a color class, all its neighbors in $G'$, and the neighbors of its neighbors in $G'$, are not activated. Consider the set of vertices $\{u_1,u_2,...,u_k\}$ selected by $v$. Any vertex $w$ in $G$ that also selected at least one of these vertices is at distance at most $2$ from $v$ in $G'$. Hence the color of $w$ is distinct from the color of $v$, and thus $w$ is not active during this round. The vertices $\{u_1,u_2,...,u_k\}$ are at distance $1$ from $v$, and thus are not active as well. The only active vertex that selected at least one of them is $v$. Hence, the memories of these $K$ vertices are assigned solely to $v$ in this time period.
\end{proof}

\begin{lemma}
\label{lemma:k-backup-rounds}
The algorithm terminates within $O({\log}^*{n} + K ^ 4)$ rounds.
\end{lemma}{}

\begin{proof}
First, \textit{$K$-next-module} is performed in parallel within a single round. Afterwards, the construction of $O(\Delta ^ 2)$ 2-hop coloring of the subgraph $G'$ of $G$ requires running time of ${\log}^*{n} + O(1)$ by \cite{linial1987distributive} \cite{linial1992locality}. Finally, the round-robin fashion of the scheduling is done for $O(\Delta'^4(G'))$ color-classes, each color-class on a different round (where each color-class performs all of the work in parallel). Yet, as the subgraph $G'$ has maximal degree of $c \cdot K$, while $c$ is a small constant and $K$ is selective, the overall round complexity of the algorithm is $O({\log}^*{n} + K ^ 4)$.
\end{proof}

The combination of Lemma \ref{lemma:k-backup-mem} and Lemma \ref{lemma:k-backup-rounds} leads to the following theorem:
\begin{theorem}
\label{theorem:k-backup}
Given a graph $G=(V,E)$, using procedure \textsc{Efficient-VM}, each active node increases its virtual memory by factor $K$. This computation requires $O({\log}^*{n} + K ^ 4)$ rounds.
\end{theorem}

\section{Extended Virtual Memory by Color Super-Classes}

In this section we devise an algorithm which supports \textit{extended} virtual memory, based on $K$-Backup Placement. In this algorithm several color classes can perform backups simultaneously during the same round. The procedure is applicable to graphs with bounded growth, in which one can color a 4-hop neighborhood with $O(\Delta)$ colors. (See, e.g., \cite{schneider2008log}.) In such graphs the neighborhood independence is bounded by a constant $c$. Moreover, here we assume a roughly uniform distribution of processors on the surface, such that the minimum degree is $\Omega(\Delta)$.

We begin with defining a phase of processing in the WSN. Suppose that just a single node $v$ in the network senses some data. The other nodes, $V \setminus v$, do not perform their own tasks, but assist $v$ with data processing. (By exchanging messages with $v$ and performing computations, but not sensing new data.) A phase consists of the rounds required to complete this task. The {\em rate of data change is denoted by the parameter $R$}. This parameter defines the number of phases during which data remains unchanged.

The procedure receives a graph $G = (V,E)$ as input and the parameter $R$, $1 \leq R < \Delta$, and proceeds as follows. We first color the network with $O(\Delta)$ colors. Then, we partition the graph using super-classes, each class consists of $O(\Delta/R)$ distinct colors, $R < \Delta$. At last, we perform $R$ phases, each with one active super-class (which is allowed to perform backup placement), in which we perform backup operation in parallel to the vertices of all the non-active super-classes (which are not allowed to perform backup placement). This completes the description of the algorithm. Its pseudocode is provided in Algorithm \ref{algo2}. The next lemmas summarizes its correctness. This algorithm is illustrated in Figure \ref{fig1}. 

\begin{algorithm}[H]
\caption{Extended Virtual Memory by Color Super-Classes}
\label{algo2}
\begin{algorithmic}[1]
\Procedure{Extended-VM(Graph $G = (V,E), \Delta, c, R$)}{}
\State \textbf{do} $(\Delta + 1)$-coloring of $G$
\State \textbf{divide} $G$ using the $(\Delta + 1)$-coloring to $\lceil \Delta + 1 / R\rceil$ $\rightarrow$ \textit{Super-Classes}
\State \textbf{foreach} \textit{sc} $\in$ \textit{Super-Classes} do: \Comment{Each vertex knows its Super-Class}
\State \hspace{0.5cm} {\bf foreach} node $v \in sc$ in parallel do: \Comment{Round-robin for Super-Classes vertices}
\State \hspace{1.0cm} $v.BP = \Gamma(v) \cap \left( \textit{Super-Classes} \setminus \textit{sc} \right) $
\State \hspace{1.0cm} \textbf{divide} the backup placement by $|v.BP|$
\State \hspace{1.0cm} \textbf{distribute} the divided backup placement in parallel from each $v \in G(sc)$ $\rightarrow$ $v.BP$ vertices
\EndProcedure
\end{algorithmic}
\end{algorithm}


\begin{figure}[H]
\centering
\begin{tikzpicture}

  \tikzstyle{vertex_r}=[circle,fill=red!25,minimum size=12pt,inner sep=2pt]
  \tikzstyle{vertex_g}=[circle,fill=green!25,minimum size=12pt,inner sep=2pt]
  \tikzstyle{vertex_b}=[circle,fill=blue!50,minimum size=12pt,inner sep=2pt]
  \tikzstyle{vertex_y}=[circle,fill=yellow!25,minimum size=12pt,inner sep=2pt]
  
  \node[vertex_b] (G_0) at (0,0)     {2};
  \node[vertex_g] (G_1) at (0.5,2)   {2};
  \node[vertex_g] (G_2) at (2,0.5)   {2};
  \node[vertex_y] (G_3) at (2,-0.5)  {1};
  \node[vertex_y] (G_4) at (0.5,-2)  {1};
  \node[vertex_g] (G_5) at (-0.5,-2) {2};
  \node[vertex_r] (G_6) at (-2,-0.5) {1};
  \node[vertex_y] (G_7) at (-2,0.5)  {1};
  \node[vertex_r] (G_8) at (-0.5,2)  {1};
  
  \draw [line width=0.2mm, black] (G_0) -- (G_1); 
  \draw [line width=0.2mm, black] (G_1) -- (G_8); 
  \draw [line width=0.2mm, black] (G_8) -- (G_0); 

  \draw [line width=0.2mm, black] (G_0) -- (G_3); 
  \draw [line width=0.2mm, black] (G_3) -- (G_2); 
  \draw [line width=0.2mm, black] (G_2) -- (G_0); 
  
  \draw [line width=0.2mm, black] (G_0) -- (G_5); 
  \draw [line width=0.2mm, black] (G_5) -- (G_4); 
  \draw [line width=0.2mm, black] (G_4) -- (G_0); 
  
  \draw [line width=0.2mm, black] (G_0) -- (G_6); 
  \draw [line width=0.2mm, black] (G_6) -- (G_7); 
  \draw [line width=0.2mm, black] (G_7) -- (G_0);

  \draw [->, line width=0.2mm, blue] [dashed] (G_8) to[out=225,in=135] (G_0);
  \draw [->, line width=0.2mm, blue] [dashed] (G_7) to[out=45,in=135] (G_0);
  \draw [->, line width=0.2mm, blue] [dashed] (G_6) to[out=315,in=225] (G_0);  
  \draw [->, line width=0.2mm, blue] [dashed] (G_4) to[out=45,in=315] (G_0);   
  \draw [->, line width=0.2mm, blue] [dashed] (G_3) to[out=225,in=315] (G_0);

\end{tikzpicture}
\caption{An extended virtual memory by color super-classes. A backup placement (\textit{blue}) in a croix pattée shaped graph with bounded neighborhood independence \textit{c = 4}, 4-coloring (red, yellow, green, blue) and 2 super-classes (1, 2), which divides the coloring into 1=\{red, yellow\} and 2=\{green, blue\}, at the 1-super-class turn to perform backup placement.}
\label{fig1}
\end{figure}
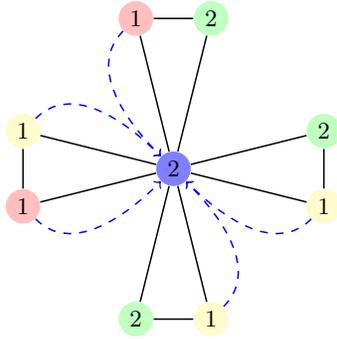

\begin{lemma}
Each vertex in $V$ is selected by at most $O(\Delta c / R)$ vertices.
\end{lemma}
\begin{proof}
Since the neighborhood independence of $G$ is $c$, there are at most $c$ neighbors colored by the same color. In each super-class there are at most $O(\Delta / R)$ colors. Hence, for each vertex there are $O(\Delta c / R)$ vertices which can select it for backup.
\end{proof}

\begin{lemma}
Each vertex $v \in V$ can select at least $deg(v) - O(\Delta c / R)$ vertices which are not active.
\end{lemma}
\begin{proof}
First we analyze the number of active neighbors of $v$. Those are neighbors which belong to the same super-class. The number of colors in the super-class is $O(\Delta / R)$. As the neighborhood independence is $c$, each color class is taken by at most $c$ neighbors, therefore the super-class contains at most $O(\Delta c /R)$ neighbors. Consequently, the amount of non-active neighbors is the overall number of neighbors minus the active neighbors, which is $\deg(v) - O(\Delta c / R)$.
\end{proof}

\begin{lemma}
The virtual memory of each vertex in $V$ is increased by $\Theta(R)$.
\end{lemma}
\begin{proof}
Let $M$ be the memory size on each vertex. A vertex can be selected by at most $O(\Delta c / R)$ vertices, thus each vertex is allocating $O(MR / \Delta c)$ virtual memory for each backuping vertex. Each backuping vertex selects $\deg(v) - O(\Delta c/R)$ vertices, which equals to $\Delta \cdot \left( 1- c/R\right)$ since the minimum degree is $\Omega(\Delta)$. Hence, the virtual memory of each vertex in $V$ is $O(M(R/c - 1))$. Since the physical memory on each node is $M$, and since $c$ is a small constant, the virtual memory is increased by $\Theta(R)$.
\end{proof}

\begin{lemma}
The algorithm terminates within $O({\log}^*{n} + R)$ rounds.
\end{lemma}{}

\begin{proof}
First, the construction of $(\Delta + 1)$-coloring of $G$ requires $O({\log}^*{n})$ rounds \cite{schneider2008log}. Then, the division of the coloring to color-classes is done in $O(1)$ rounds using division of colors into super-classes. Finally, the round-robin fashion of the backup placement is done for $O(R)$ rounds, where each super-class performs all of the work in parallel. Therefore, the running time of the algorithm is $O({\log}^*{n} + R)$.
\end{proof}

As a corollary of the lemmas above, the following theorem is obtained.
\begin{theorem}
Given a constant $R$, and a uniformly distributed WSN forming a graph $G=(V,E)$ with bounded growth, and bounded neighborhood independence $c$, using \textsc{Extended-VM} each node increases its virtual memory by $\Theta(R)$. This computation requires $O({\log}^*{n} + R)$ rounds.
\end{theorem}

\section*{Acknowledgments}
The authors are grateful to the anonymous reviewers of \textit{ALGOSENSORS} 2020 for their very helpful suggestions which helped us greatly improve this paper. The authors also would like to thank Harel Levin for his fruitful comments and extensive evaluation of this work.
This work was supported by the Lynn and William Frankel Center for Computer Science, the Open University of Israel's Research Fund, and ISF grant 724/15. 

\bibliographystyle{plain}
\bibliography{bibliography}
\end{document}